\documentclass[a4paper,UKenglish,cleveref,autoref]{lipics-v2021}

\newtheorem*{theorem*}{Theorem}

\title{Decomposing Directed Graphs into Partial Orders is Hard}

\author{Grzegorz Gutowski}
  {Institute of Theoretical Computer Science, Faculty of Mathematics and Computer Science, Jagiellonian University, Krak{\'o}w, Poland \and \url{https://grzegorz.gutowscy.pl}}
  {grzegorz.gutowski@uj.edu.pl}{https://orcid.org/0000-0003-3313-1237}
  {Partially supported by grant no.~2023/49/B/ST6/01738 from National Science Centre, Poland.}
\author{Dominik Jilg}
  {Julius-Maximilians-Universit\"at W\"urzburg, Germany}
  {dominik.jilg@uni-wuerzburg.de}{https://orcid.org/0009-0004-4527-9583}{}
\author{Antonio Lauerbach}
  {Julius-Maximilians-Universit\"at W\"urzburg, Germany}
  {lauerbach@informatik.uni-wuerzburg.de}{https://orcid.org/0009-0007-9093-3443}{}
\author{Miko{\l}aj Rams}
  {Institute of Theoretical Computer Science, Faculty of Mathematics and Computer Science, Jagiellonian University, Krak{\'o}w, Poland}
  {mikolaj.rams@student.uj.edu.pl}{https://orcid.org/0009-0004-9142-9167}{}

\authorrunning{G.~Gutowski, D.~Jilg, A.~Lauerbach, and M.~Rams}

\Copyright{Grzegorz Gutowski, Dominik Jilg, Antonio Lauerbach, and Miko{\l}aj Rams}

\hideLIPIcs
\nolinenumbers

\usepackage[T1]{fontenc}
\usepackage{booktabs}
\usepackage{mathtools}
\usepackage{array}
\usepackage{siunitx}
\usepackage{xcolor}
\usepackage[size=footnotesize]{todonotes}
\usepackage{xspace}
\usepackage[disableredefinitions]{complexity}
\usepackage{gensymb}
\usepackage{cite}
\usepackage{pgfplots}
\usepackage{tikz}
\usepackage{amsmath}
\usepackage{amssymb}
\usepackage{tcolorbox}
\usepackage{caption}
\usepackage{apptools}
\usepackage{lineno}

\crefname{problem}{Problem}{Problems}
\Crefname{problem}{Problem}{Problems}

\newcommand{\dummy}[1]{}

\newcommand{\rbarc}[1][r]{\ensuremath{#1}{-rainbow arc set}\xspace}
\newcommand{\podcrit}[1][k]{\ensuremath{#1}{-pop-critical}\xspace}

\newcommand{\linedg}[1]{\vec{L}\left(#1 \right)\xspace}
\newcommand{\podn}[1]{\mathrm{pod}\left(#1 \right)\xspace}
\newcommand{\popn}[1]{\mathrm{pop}\left(#1 \right)\xspace}
\DeclareMathOperator{\fin}{in}
\DeclareMathOperator{\fout}{out}

\pgfplotsset{compat=1.18,small,scale only axis}
\graphicspath{{./figures/}}

\definecolor{dark blue}{rgb}{0.121,0.47,0.705}
\let\emph\relax\DeclareTextFontCommand{\emph}{\color{dark blue}\em}

\newtheorem*{observation*}{Observation}

\begin{document}
\maketitle

\begin{abstract}
  Comparability graphs represent partial orders and provide natural building blocks for representing a graph as a union of simpler graphs.
  The analogous task in the directed setting is to represent a digraph as a union of partial orders.
  Aboulker et al.~\cite{aboulker2026} recently introduced the \emph{partial order decomposition number} as a useful structural parameter in their study of tournaments.
  Formally, a partial order decomposition of a digraph~$D$ is a cover of its arc set by transitive acyclic subdigraphs, each representing a strict partial order on the vertex set of~$D$.
  The partial order decomposition number is the minimum number of such subdigraphs needed.
  We also consider the partitioning variant, in which every arc belongs to exactly one of the partial orders in the decomposition.
  We prove that, for every fixed~$k\geq 2$, deciding whether a digraph admits a decomposition into at most~$k$ partial orders is \NP-complete, even for directed acyclic graphs (DAGs).
  For~$k=2$, we reduce from \textsc{NAE-3-Sat}.
  For~$k\geq 3$, we reduce from \textsc{$k$-Coloring}, and the reduction yields DAGs with triangle-free underlying graphs.
  Both reductions establish hardness for both variants.
  We further show that the minimum cover size and the minimum partition size can differ even on DAGs.
\end{abstract}

\section{Introduction}
Representing a graph as a union of simpler graphs is a standard way to measure its structural complexity.
Different choices of which graphs are considered simpler lead to different interesting classifications.
For the purpose of this paper, we say that a graph is simple if it is a comparability graph. 
A \emph{comparability graph} has the elements of a partially ordered set as its vertices, with two distinct vertices adjacent if and only if they are comparable.
Equivalently, it is a graph that admits a transitive acyclic orientation.
Comparability graphs are perfect: in every induced subgraph, the chromatic number equals the clique number.
In particular, the chromatic number of a comparability graph equals the clique number.
Consequently, if the edge set of a graph~$G$ is covered by $r$ comparability subgraphs, then we get $\chi(G)\leq\omega(G)^r$ from the product coloring obtained from proper colorings of the covering graphs.
This principle has geometric applications.
For example, Larman, Matou\v{s}ek, Pach, and T\"{o}r\H{o}csik~\cite{larman1994ramsey} cover the disjointness graph of a family of planar convex sets by four comparability graphs.
More recently, Gy\'arf\'as, Marits, and T\'oth~\cite{gyarfas2024partitioning} studied covers and edge partitions into comparability graphs systematically.
%They showed that two parts suffice for almost all perfect graphs, whereas even interval graphs can require arbitrarily many parts.

In the directed setting, the orientations of the edges are prescribed, and we ask for a decomposition into transitive acyclic digraphs.
We identify such a digraph with a strict partial order on its vertex set.
Following Aboulker, Crew, Duron, Fan, Jacob, Kimbrough, Koerts, Moore, Spirkl, and Thomassé~\cite{aboulker2026}, a \emph{partial order decomposition} of a digraph~$D$ is a family of partial orders $Q_1,\ldots,Q_k$ on $V(D)$ such that
\[
    E(D)=E(Q_1)\cup\cdots\cup E(Q_k).
\]
The minimum possible number of parts is the \emph{partial order decomposition number}, denoted by~$\podn{D}$.
We also consider the partitioning variant, in which the arc sets of the partial orders must be pairwise disjoint, and define the corresponding \emph{partial order partition number}, denoted by~$\popn{D}$.
%Equivalently, this variant asks for an arc coloring in which each color class is a partial order.
For a fixed integer~$k$, the associated decision problems ask whether $\podn{D}\leq k$ or $\popn{D}\leq k$, respectively.

The structural motivation for this parameter comes from directed coloring.
Aboulker et al.~\cite{aboulker2026} prove that
\[
    \vec{\chi}(D)\leq\vec{\omega}(D)^{\podn{D}},
\]
where $\vec{\chi}(D)$ is the minimum number of acyclic induced subdigraphs partitioning $V(D)$, and $\vec{\omega}(D)$ is the minimum clique number of a backedge graph over all vertex orderings of~$D$.
Here, a backedge graph has an undirected edge for each arc directed backwards in the chosen ordering.
Thus, bounding the number of partial orders gives a polynomial bound on the dichromatic number in terms of the directed clique number.
The same work asks about the complexity of computing or approximating $\podn{T}$ for a tournament~$T$.
These results motivate the algorithmic study of partial order decompositions on general digraphs and on restricted classes.

Directed acyclic graphs (DAGs) are a particularly natural class to consider.
Every subdigraph of a DAG is already acyclic, so the decomposition problem concerns only transitivity.
Nevertheless, transitivity imposes constraints between different arcs: if a part contains $u\to v$ and $v\to w$, it must also contain $u\to w$.
In particular, if the latter arc is absent in the input graph, the two former arcs cannot belong to the same part.
Moreover, deleting an arc from a partial order can destroy transitivity, so assigning each covered arc to just one of its covering parts need not produce a valid partition.
The covering and partitioning variants therefore require separate consideration.
In \cref{fig:different}, we present a simple example of a digraph~$D$ that has a unique covering with two partial orders.
To see that the covering is unique, it is enough to observe that any two arcs that meet head-to-tail, $u \rightarrow v \rightarrow w$, with the arc $u \rightarrow w$ not present in $D$, cannot be in the same partial order.
In the unique covering, the arc in the center of the figure is contained in both partial orders, so the digraph cannot be partitioned into two partial orders, and we have $\podn{D}=2$ and $\popn{D}=3$.

\begin{figure}[tbp]
   \centering
   \includegraphics[page=3]{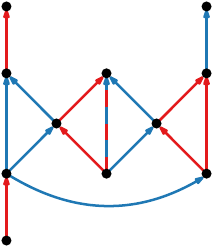}
   \caption{A DAG~$D$ with $\podn{D}=2$ and $\popn{D}=3$. The graph has a unique covering with two partial orders (depicted blue and red). Since the central arc must be contained in both partial orders, the graph cannot be partitioned into two partial orders.}
   \label{fig:different}
\end{figure}

\paragraph*{Our results.}
We establish the following complexity classification for both variants.
\begin{theorem*}
    For every fixed integer $k\geq2$, deciding whether the arcs of a DAG can be covered by at most~$k$ partial orders is \NP-complete, as is deciding whether they can be partitioned into at most~$k$ partial orders.
    For $k\geq3$, both statements hold even when the underlying graph is triangle-free.
\end{theorem*}
For $k=1$, both problems are polynomial-time solvable by checking acyclicity and transitivity.
Our reduction for $k=2$ is from \textsc{NAE-3-Sat} and applies simultaneously to covering and partitioning.
It shows that allowing transitive triangles makes the two-part problem hard even on DAGs, whereas the triangle-free case is polynomial-time solvable.
For $k\geq3$, we present a reduction from \textsc{$k$-Coloring} using partial-order-critical DAGs.
In \cref{sec:k2}, we present the reduction for $k=2$, and in \cref{sec:many-hardness}, we develop the partial-order-critical digraphs and show the reduction for $k\geq3$.
In \cref{sec:open}, we discuss remaining open problems.

% \paragraph{POD and POP are Different}
% Consider the DAG~$D$ shown in \cref{fig:different}. Due to missing transitive arcs, it has a unique decomposition into two partial orders. In this decomposition, the middle arc is forced to be in both partial orders. Therefore, the graph cannot be partitioned into two partial orders, and we have $\podn{D}=2$ and $\popn{D}=3$.

% \section{Preliminaries}
% \label{sec:preliminaries}
% \begin{itemize}
%     \item partial order 
%     \item Formal definition of Covering/ Coloring Variant and Pod number (minimizing/ deciscion variant of the problem -> introduction)
%     % \item NAE-3-SAT definition (maybe use triples insteat of conjunctions)
%     % \item chromatic number
%     % \item k-COLORING
%     % \item DAGs
%     % \item Underlying linegraph / line digraph / $UG(G)$ / $L(G)$
%     % \item pop-critical graphs
% \end{itemize}

\section{Hardness Result for \texorpdfstring{$k=2$}{k=2}}
    \label{sec:k2}
    This section is devoted to the following theorem.
    \begin{theorem}
        \label{thm:k2}
        The problems of deciding whether $\podn{G}\le2$ and 
        $\popn{G}\le2$ are both \NP-complete, even when $G$ is 
        restricted to be a directed acyclic graph.
    \end{theorem}

    We reduce from \textsc{Not-All-Equal-3-Sat}, abbreviated as \textsc{NAE-3-Sat}.
    An instance $\Phi=( X, C)$ consists of a set $X=\{x_1,\ldots,x_n\}$ of variables and 
    a family $C=(C_1,\ldots,C_m)$ of triples of literals, where
    $C_j=(\ell_j^1,\ell_j^2,\ell_j^3)$. We call these triples clauses.
    A truth assignment $\varphi\colon X\to\{0,1\}$ extends to literals by
    $\varphi(\lnot x_i)=1-\varphi(x_i)$. It satisfies a clause if the clause
    contains both a true and a false literal, and it satisfies $\Phi$ if
    it satisfies every clause. Equivalently, we require
\[
    \{\varphi(\ell_j^1),\varphi(\ell_j^2),\varphi(\ell_j^3)\}
    =\{0,1\}
    \qquad \forall~C_j,~j\in\{1,\ldots,m\}.
\]
The problem asks whether such an assignment exists. We use the
\NP-complete variant in which each clause contains literals over
three distinct variables~\cite{schaefer_complexity_1978}.
The order within each triple is arbitrary but fixed for the construction.

\begin{proof}[Proof of \Cref{thm:k2}]
    Both problems are clearly in \NP\ since we can verify a solution in
    polynomial time. 
    Let $\Phi=(X, C)$ be an instance of \textsc{NAE-3-Sat} as
    defined above. We construct a directed acyclic graph $G$ as follows.
    The same construction will work for both the covering and the partitioning variants, as will be discussed later in the proof.
    For each variable~$x_i$, create a variable gadget $\mathcal{X}_i$ that consists of a directed 
    path $x_i^0 \rightarrow x_i^1 \rightarrow x_i^2$. Since the shortcut $x_i^0 \rightarrow  x_i^2$ is absent, 
    transitivity prevents the two arcs of this path from belonging to the same partial order. Consequently, 
    in every decomposition into two partial orders, they belong to different partial orders. The variable gadget therefore 
    has two possible states, corresponding to the two possible truth values of $x_i$.
    For each clause $C_j$, we create a clause gadget ${\cal C}_j$ that consists of a 
    directed path $c_j^1 \rightarrow c_j^2 \rightarrow c_j^3 \rightarrow c_j^4$ 
    together with the shortcut arcs $c_j^1 \rightarrow  c_j^3, c_j^2 \rightarrow  c_j^4$.
    In particular, the arc $c_j^1 \rightarrow  c_j^4$ is not present. The three consecutive path arcs
    will encode the truth values of the three literals in each clause. In particular, 
    they cannot all belong to the same partial order, since transitivity
    would then require the existence of $c_j^1 \rightarrow c_j^4$. The gadgets are shown in \cref{fig:gadgets}.
    \begin{figure}[ht]
        \centering
        \includegraphics[page=1, width=\textwidth]{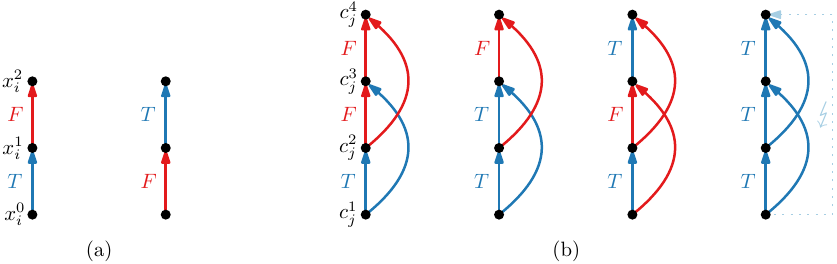}
        \caption{The gadgets used in the reduction. (a) The two states of a variable gadget. 
        (b) Representative colorings of a clause gadget. The three consecutive path arcs encode the truth values of the
        respective literals in the clause. A monochromatic clause path would require the missing shortcut $c_j^1 \rightarrow c_j^4$, shown dotted.}
        \label{fig:gadgets}
    \end{figure}
    We connect the gadgets as follows: for a literal $\ell$, define its attachment vertex $a(\ell)$ by
    $a(\ell) = x_i^1$ if $\ell = x_i$ and $a(\ell) = x_i^2$ if $\ell = \lnot x_i$. We also write
    $e(\ell)$ for the variable-gadget arc entering $a(\ell)$.
    That is, for $\ell = x_i$ we have $e(\ell) = x_i^0 \rightarrow  x_i^1$ and for $\ell = \lnot x_i$ we have $e(\ell) = x_i^1 \rightarrow  x_i^2$.
    For each clause $C_j$, we add three \emph{connecting arcs} $a(\ell_j^1) \rightarrow c_j^1$, $a(\ell_j^2) \rightarrow c_j^2$, and $a(\ell_j^3) \rightarrow c_j^3$. %To transmit the truth values of the literals to the clause gadgets
    Furthermore, we add three \emph{guard arcs} $a(\ell_j^1) \rightarrow c_j^3$, $a(\ell_j^1) \rightarrow c_j^4$, and $a(\ell_j^2) \rightarrow c_j^4$.
    An example is shown in \cref{fig:phi}.

    \begin{figure}[ht]
        \centering
        \includegraphics[page=2, width=\textwidth]{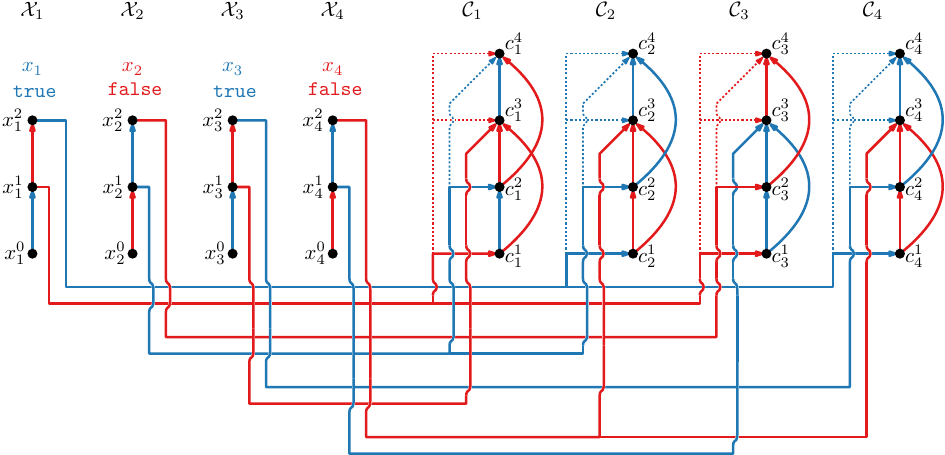}
        \caption{The construction with clauses $(x_1,x_2,x_3)$, $(\lnot x_1, x_2,\lnot x_4)$, $(x_1,\lnot x_2,x_4)$, $(\lnot x_1,\lnot x_3,\lnot x_4)$. The coloring corresponds to the NAE assignment $x_1=x_3=1$ and $x_2=x_4=0$. Guard arcs are shown dotted.}
        \label{fig:phi}
    \end{figure}
    Note that $G$ contains exactly $3n +4m$ vertices and $2n +11m$ arcs. Clearly, $G$ can be constructed from $\Phi$ in polynomial time. Note also that $G$ contains no cycles. We also observe that, on
    the constructed instances, a cover by two partial orders exists if and only if a partition into two partial orders exists. Indeed, every arc of $G$ belongs to a
    directed $2$-path whose shortcut is absent. This is obvious for the arcs belonging
    to variable gadgets. Each connecting or guard arc leaving $a(\ell_j^r)$ forms such a path with
    $e(\ell_j^r)$. Each clause-path arc $c_j^r \rightarrow c_j^{r+1}$ forms such a path with
    $a(\ell_j^r) \rightarrow c_j^r$. Finally, each of the two clause-gadget shortcuts lies on a $2$-path 
    from $c_j^1$ to $c_j^4$, whose shortcut is absent. Consequently, $G$ admits a 
    cover by two partial orders if and only if it admits a partition into two partial orders.
    We may therefore treat both variants as an arc coloring with colors blue and red in which
    each color class is transitive. The preceding observations also imply that, in every valid such coloring,
    all connecting and guard arcs leaving $a(\ell_j^r)$ have the color opposite to $e(\ell_j^r)$, while
    $c_j^r \rightarrow c_j^{r+1}$ has the same color as $e(\ell_j^r)$.

    Suppose first that $\varphi$ satisfies $\Phi$.
    If $x_i$ is true, color $x_i^0 \rightarrow  x_i^1$ blue and $x_i^1 \rightarrow  x_i^2$ red; 
    otherwise reverse the colors. Thus, $e(\ell)$ is blue if 
    and only if $\varphi(\ell)=1$. 
    Color every connecting and guard arc leaving $a(\ell_j^r)$ opposite to $e(\ell_j^r)$,
    and color the clause-path arc $c_j^r \rightarrow c_j^{r+1}$ with the same color as $e(\ell_j^r)$.
    Finally, color $c_j^1 \rightarrow  c_j^3$ opposite to $c_j^3 \rightarrow  c_j^4$ and 
    $c_j^2 \rightarrow  c_j^4$ opposite to $c_j^1 \rightarrow  c_j^2$. 

    We verify transitivity by considering monochromatic $2$-paths. 
    No such path lies within a variable gadget or consists of a variable 
    arc followed by a connecting or guard arc. Within a clause gadget, if
    two path arcs have the same color, the third has
    the opposite color because $\varphi$ satisfies the clause. Moreover,
    all $2$-paths within the clause gadget involving the shortcuts $c_j^1 \rightarrow  c_j^3$ and
    $c_j^2 \rightarrow  c_j^4$ are not monochromatic by definition.
    For every connecting arc $a(\ell_j^r) \rightarrow c_j^r$ followed by an arc $c_j^r \rightarrow  c_j^{r+1}$ of the clause path, the two arcs have opposite colors, so the $2$-path is not monochromatic.
    Furthermore, there are two cases where a connecting arc is followed by a shortcut arc of the clause gadget: if~$a(\ell_j^1) \rightarrow c_j^1$ and~$c_j^1 \rightarrow c_j^3$ are monochromatic, then the transitive arc~$a(\ell_j^1)\rightarrow c_j^3$ must also have the same color, which is guaranteed by the coloring of the guard arcs. Similarly, if~$a(\ell_j^2) \rightarrow c_j^2$ and~$c_j^2 \rightarrow c_j^4$ are monochromatic, then the transitive guard arc~$a(\ell_j^2) \rightarrow c_j^4$ has the same color.
    Finally, if~$a(\ell_j^1) \rightarrow c_j^3$ and~$c_j^3 \rightarrow c_j^4$ are monochromatic, then the transitive guard arc~$a(\ell_j^1) \rightarrow c_j^4$ also has the same color.
    % Furthermore, all connecting arcs followed by a path arc of a clause gadget also form $2$-paths that are not monochromatic. 
    % Finally, a connecting arc followed by a shortcut of a clause gadget is either not monochromatic, or if it is, then the shortcut is one of the guard arcs and colored with the same color.
    % Lastly, a guard arc may be followed by a path arc of the same color. However, then there is also a guard arc as the transitive arc, with the same color.
        
    Conversely, suppose that $G$ admits a valid blue-red arc
    coloring. Define $\varphi(x_i)=1$ if and only if $x_i^0 \rightarrow x_i^1$
    is blue. The two variable-gadget arcs have opposite colors, so
    $e(\ell)$ is blue if and only if $\varphi(\ell)=1$. By the forced
    color relations established above, the three clause-path arcs encode
    the values of $\ell_j^1,\ell_j^2,\ell_j^3$. They cannot all have
    the same color, since transitivity would then force the missing arc
    $c_j^1 \rightarrow c_j^4$. Hence, each clause contains both a true and a false
    literal, and $\varphi$ satisfies $\Phi$.
    \end{proof}

We were able to use a similarly constructed reduction to prove \NP-hardness for $k=3$, but not for $k \geq 4$. Two issues occur when trying this approach for $k=4$: gadgets are difficult to find, and the relation between the decomposition and the partition (coloring) variants becomes unclear.
To address these problems, we present a completely different reduction that proves hardness for $k \geq 3$, but does not seem to be applicable in the case $k=2$.

\section{Hardness Result for \texorpdfstring{$k\geq 3$}{k >= 3}}
\label{sec:many-hardness}

The main theorem of this section is the following one.

\begin{theorem}
  \label{thm:k3}
  For every fixed integer $k\geq3$,
  the problems of deciding whether $\podn{G}\leq k$ and 
    $\popn{G}\leq k$ are both \NP-complete, even when $G$ is 
    restricted to be a directed acyclic graph whose underlying graph is
    triangle-free.
\end{theorem}

In the following, we start by investigating some properties of both the partial order decomposition and partition numbers and their relation to colorings of line digraphs. We further introduce \podcrit digraphs and the properties needed for our reduction.

\subsection{Structural Properties}

The \emph{underlying graph} $UG(D)$ of a digraph $D$ is defined as a graph with $V(UG(D)) = V(D)$ and $vw \in E(UG(D))$ if $(v \rightarrow  w) \in E(D)$ or $(w \rightarrow  v) \in E(D)$.
Given a digraph $D = (V,E)$, the \emph{line digraph} $\linedg{D}$ has vertex set $E$ and arc set $\left\{ (u \rightarrow  v, v \rightarrow  w) : (u \rightarrow  v), (v \rightarrow  w) \in E \right\}$.
Informally, whenever two arcs of $D$ meet head-to-tail, there is an arc between them in the natural direction in $\linedg{D}$.
We will use the following result on line digraphs.
\begin{theorem}[Erd\H{o}s and Hajnal~\cite{erdoshajnal}, Harner and Entringer \cite{harner1972arc}]\label{thm:harner}
    For any digraph $D$, we have 
    \[ \chi\left(UG(\linedg{D})\right) \geq \log_2(\chi(UG(D)))\text{.} \]
\end{theorem}

For convenience, we call a digraph $D$ a triangle-free DAG if it is a directed acyclic graph whose underlying undirected graph is triangle-free.
As in the reduction given in \cref{sec:k2}, we refer to partial orders in a partition as color classes. We abuse notation and denote by $c(E')$ the set of colors used by the coloring $c$ on a subset of arcs $E' \subseteq E(D)$.

As described in the introduction, the partial order decomposition number and the partial order partition number can differ even for DAGs. As the following lemma shows, this separation is not possible for triangle-free DAGs.

\begin{lemma}
    \label{lem:podpopequal}
    For every triangle-free DAG $D$, we have $\podn{D} = \popn{D}$.
\end{lemma}

\begin{proof}
    Let $D$ be a triangle-free DAG. Since every partition of $D$ is also a cover, we have $\podn{D} \leq \popn{D}$.
    Let $k = \podn{D}$ and let $Q_1 \cup \ldots \cup Q_k$ be a decomposition of $D$. Let $u \rightarrow  v$ be an arc in $Q_i$. 
    Since $UG(Q_i)$ is triangle-free, there is no $w$ such that both $u \rightarrow  w, w \rightarrow  v$ are in $Q_i$.
    It follows that $Q_i \setminus (u \rightarrow v)$ remains a partial order.
    We can therefore delete an arc covered by multiple partial orders from any of them.
    This allows us to transform the decomposition into a partition, so $\popn{D} \leq \podn{D}$.
\end{proof}

This lemma allows us to show hardness for both variants while only considering the partition variant.
We continue with a lemma that relates the partial order partition number to the chromatic number of the underlying graph of the line digraph. Before stating the lemma, we give a necessary and sufficient condition for $Q$ to be a partial order. By definition, a DAG $Q$ is a partial order if and only if \[(u \rightarrow  v), (v \rightarrow  w) \in E(Q) \implies (u \rightarrow  w) \in E(Q)\text{.}\]
We can also characterize partial orders of DAGs in terms of forbidden induced subgraphs.
Note that if two arcs $u\rightarrow v$ and
$v \rightarrow w$ meet head-to-tail in a DAG, the vertices $u,v,w$ induce either a directed path $P_3$ or a transitive tournament~$TT_3$.
Thus, a DAG is a partial order if and only if it contains no $P_3$ as an induced subgraph.

\begin{lemma}
    Let $D$ be a triangle-free DAG and $Q \subseteq E(D)$. Then, $Q$ is a partial order if and only if $Q$ is an independent set in $UG(\linedg{D})$.
\end{lemma}

\begin{proof}
    $(\Rightarrow)$ Since $Q$ is a partial order, there is no induced~$P_3$, as a~$P_3$ violates transitivity.
    Since $UG(Q)$ is triangle-free, there is also no $TT_3$ in $Q$. Therefore, $Q$ does not contain two arcs meeting head-to-tail, so $Q$ is an independent set in $UG(\linedg{D})$.\\
    $(\Leftarrow)$ Since $Q$ is an independent set in $UG(\linedg{D})$, $Q$ does not contain a three-vertex directed path as a subgraph, which proves the claim.
\end{proof}

A proper coloring of $UG(\linedg{D})$ is a partition of $V(UG(\linedg{D}))$ into independent sets, which leads to the following conclusion.

\begin{corollary}
    \label{cor:chromaticpoppod}
    For every triangle-free DAG $D$, we have \[\popn{D} = \podn{D} = \chi(UG(\vec{L}(D)))\text{.}\]
\end{corollary}

We can now use properties of proper colorings to investigate the partial order partition number in triangle-free DAGs.

\begin{lemma}
    \label{lem:decreasepop}
    Let $D$ be a triangle-free DAG, let $e \in E(D)$ and let $D' = D \setminus e$. Then $D'$ is a triangle-free DAG and $\popn{D}-1 \leq \popn{D'} \leq \popn{D}$.
\end{lemma}

\begin{proof}
    We first observe that $D'$ remains acyclic and $UG(D')$ remains triangle-free after the removal of an arc $e \in E(D)$. It follows that $D'$ is a triangle-free DAG.
    By \cref{cor:chromaticpoppod}, we have $\chi(UG(\linedg{D})) = \popn{D}$, and therefore removing a vertex from $UG(\linedg{D})$ can decrease its chromatic number by at most $1$ while not increasing it.
\end{proof}

The authors of \cite{aboulker2026} already showed that there exist triangle-free DAGs with an arbitrarily large partial order decomposition number.
We include a proof for completeness.

\begin{proposition}
    \label{lem:large-pod-triangle-free}
    For every integer $k$ there exists a triangle-free DAG $D$ that has $\podn{D} \geq k$.
\end{proposition}

\begin{proof}
    Let $G$ be a triangle-free undirected graph with $\chi(G) \geq 2^k$, which can be obtained using a known construction from the literature (e.g., \cite{Mycielski1955}).
    We set $D$ to be an acyclic orientation of $G$. By \cref{thm:harner} and \cref{cor:chromaticpoppod}, it follows that
    \[k \leq \log_2(\chi(UG(D))) \leq \chi\left(UG(\vec{L}(D))\right) = \popn{D} = \podn{D},\]
    concluding the proof.
\end{proof}

We proceed by defining \podcrit digraphs and 
\rbarc{}s. We then show that there exist triangle-free DAGs that are \podcrit and how to use this property to create a digraph that contains a \rbarc.

\begin{definition}[\podcrit digraph]
    \label{def:crit}
    A digraph $D$ is \podcrit if and only if:
    \begin{enumerate}
        \item $\popn{D} = k$ and
        \item for every arc $e \in E(D)$: $\popn{D \setminus e} = k-1$.
    \end{enumerate}
\end{definition}

\begin{lemma}
    \label{lem:podcrit-triangle-free}
    For every integer $k\geq 1$ there exists a triangle-free DAG that is \podcrit.
\end{lemma}

\begin{proof}
    % We start with an arbitrary triangle-free DAG $D$ with $\popn{D} \geq k$. 
    % Since a digraph without edges has a partial order partition number of $1$, and removing an edge can decrease the partial order partition number by at most $1$, see \cref{lem:decreasepop}, we can iteratively remove edges from $D$ until we reach a $k$-pop-critical digraph.
    We start with an arbitrary triangle-free DAG $D$ with $\popn{D}\geq k$, guaranteed to exist by \cref{lem:large-pod-triangle-free}, and iteratively apply the following procedure:
    \begin{enumerate}
        \item If there exists an arc $e \in E(D)$ such that $\popn{D \setminus e} \geq k$, we remove $e$.
        \item Otherwise, return $D$.
    \end{enumerate}
    Let $D'$ be a digraph returned by this procedure. 
    %A digraph with no edges has a partial order partition number of $0$.
    As shown in \cref{lem:decreasepop}, at each step of the procedure, the partial order partition number decreases by at most $1$, so $\popn{D'} = k$. Thus,~$D'$ is a \podcrit triangle-free digraph.
\end{proof}

We now define $r$-rainbow arc sets, which will be crucial for our reduction. We also show how to construct them using pop-critical digraphs.

\begin{definition}[\rbarc]
    Let $D$ be a digraph and $S\subseteq E(D)$. The subset~$S$ is a \rbarc if and only if:
    \begin{enumerate}
        \item for any $(r+1)$-coloring $c$ corresponding to a partial order partition of $D$, we have $\left| c(S) \right| \geq r$ and
        \item there exists such an $(r+1)$-coloring $c'$ such that $\left| c'(S) \right| = r$.
    \end{enumerate}
\end{definition}

% As the following lemma shows, we can construct a \rbarc by removing a single arc from a \podcrit[(k+1)] digraph.
% To construct a digraph containing a \rbarc, we take a \podcrit[(k+1)] digraph and remove one arc that is incoming to a sink.
% We construct a digraph containing a \rbarc by taking a \podcrit[(k+1)] digraph and removing any arc that is incoming to a sink.

In the following, we write $\fin(v)$ and $\fout(v)$ for the sets of arcs entering and leaving a vertex $v$, respectively.
In order to construct a \rbarc[(k-1)], we make use of the following lemma.

\begin{lemma}
    \label{lem:u-sink}
    Let $D$ be a \podcrit[(k+1)] triangle-free DAG, let $s \in V(D)$ be a sink, let $(u \rightarrow  s) \in \fin(s)$ and let $D' = D \setminus (u \rightarrow s)$.
    Let $c$ be a $k$-coloring of the arcs of $D'$ corresponding to a partial order partition of $D'$. Then
    \begin{enumerate}
        \item $|c(\fin(u))| = k$.
        \item $u$ is a sink in $D'$.
    \end{enumerate}
\end{lemma}
\begin{proof}
    By \cref{cor:chromaticpoppod}, we have that~$c$ is also a proper coloring of $UG(\vec{L}(D'))$. 
    In~$UG(\vec{L}(D))$, the arc~$u \rightarrow  s$ is only connected to the arcs in~$\fin(u)$, since~$\fout(s) = \emptyset$. Thus, if~$|c(\fin(u))| < k$, we could color~$u \rightarrow  s$ with one of the $k$ colors not used in~$c(\fin(u))$, which would yield a proper coloring of~$UG(\vec{L}(D))$ with $k$ colors. This contradicts the assumption that $D$ is \podcrit[(k+1)].

    For the second statement, assume there exists an arc $u \rightarrow  v$ in $D'$. Since $|c(\fin(u))| = k$ and $u \rightarrow  v$ is connected to all arcs from $\fin(u)$ in $UG(\vec{L}(D'))$, we cannot color it with any of these $k$ colors. This contradicts $\popn{D'}=k$, which is implied by $D$ being \podcrit[(k+1)] and \cref{lem:decreasepop}.
    % By \cref{lem:decreasepop} we have that $c$ is a proper coloring in $UG(\vec{L}(D'))$.
    % If we reinsert $u \rightarrow  s$ to $UG(\vec{L}(D'))$ it would be connected only to $\fin(u)$, since $\fout(s) = \emptyset$.
    % If $|c(\fin(u))| < k$ we could color $UG(\vec{L}(D'))$ with $c$ and color $u \rightarrow  s$ with a color missing in $c(\fin(u))$. This would yield a coloring of $UG(\vec{L}D)$ with $k$ colors, contradicting the assumption that $D$ is \podcrit[k+1]. This proves the first statement.
    % Assume there exists an arc $u \rightarrow  v$ in $D'$. Since $|c(\fin(u))| = k$ and $u \rightarrow  v$ is connected to all arcs from $\fin(u)$ in $UG(\vec{L}(D'))$, we cannot color it with any of these $k$ colors. This contradicts $D$ being \podcrit[k+1], which proves the second statement.
\end{proof}

% We use the previous cosiderations to find a $k$-rainbow arc sets.
The following lemma shows that we can construct a \rbarc[(k-1)] on triangle-free DAGs by removing a single arc from a \podcrit[(k+1)] triangle-free DAG.

\begin{lemma}
    \label{lem:rainbowarc-set-exists}
    Let $D$ be a \podcrit[(k+1)] triangle-free DAG, let $s \in V(D)$ be a sink, let $u \rightarrow  s$ be an arc to the sink, and let~$D' = D \setminus (u \rightarrow s)$. There exists a \rbarc[(k-1)] $S \subseteq \fin(u)$ in~$D'$.
\end{lemma}
    %  and let $D' = D \setminus u \rightarrow  s$. Then there 
\begin{proof}
    We start by choosing $S' \subseteq \fin(u)$ with the following properties:
    \begin{enumerate}
        \item In every $k$-coloring $c$ corresponding to a partial order partition of $D'$, we have $\left| c(S') \right| = k$ and
        \item $S'$ is a minimal set with that property.
    \end{enumerate}
    % In order to show that such a set~$S'$ exists, we first show that in every $k$-coloring~$c$ of~$D'$, it holds that~$|c(\fin(u))| = k$. 
    % Let~$c$ be a $k$-coloring of $D'$. By \cref{cor:chromaticpoppod}, it follows that~$c$ is a proper coloring of~$UG(\linedg{D'})$. 

    Such a set $S'$ exists, since $\fin(u)$ satisfies the first property by \cref{lem:u-sink}, and we can therefore choose a subset of $\fin(u)$ that is minimal under inclusion and has the desired property.
    Let $e \in S'$, and let $S = S' \setminus e$. For any $k$-coloring $c$, we have $|c(S)|\geq k-1$. Since $S'$ was chosen to be minimal, there exists a $k$-coloring $c'$ of $D'$ with $|c'(S)| < k$.
\end{proof}

\begin{figure}[ht]
    \centering
    \includegraphics[page=4]{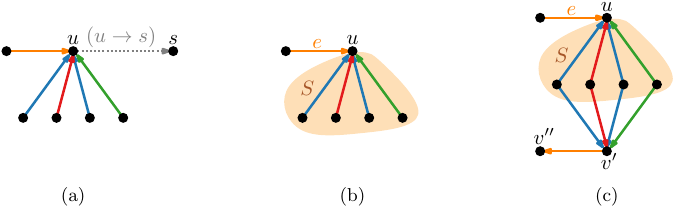}
    \caption{A part of $\mathcal V_v$ for $k=4$. (a) Deletion of $(u \rightarrow s)$ from a $5$-pop-critical DAG. (b) Selection of the $3$-rainbow arc set $S$ (shaded). (c) Copy of $S$ to $v'$, leaving one color for $(v' \rightarrow v'')$, which is the same color as of edge $e$ excluded from~$S$ in the construction of the $3$-rainbow arc set.}
    \label{fig:rainbow}
\end{figure}

\subsection{Reduction}

We are now ready to prove the main theorem of the section.
%\todo{al: Remove proof environment, write the reduction in plain text and make the claim a lemma with its own proof environment}
\begin{proof}[Proof of \Cref{thm:k3}]
    Both problems are clearly in \NP\ since we can verify a solution in polynomial time. To prove \NP-hardness of both variants, by \cref{lem:podpopequal}, it is sufficient to show hardness for partial order partitions, as we restrict our problem to triangle-free DAGs. We reduce from $k$-\textsc{Coloring}, which is \NP-hard for $k \geq 3$ even when restricted to triangle-free graphs \cite{maffray1996}.

    Let $G=(V,E)$ be a triangle-free instance of $k$-\textsc{Coloring}. For each $v\in V$, we create a vertex gadget $\mathcal V_v$ as follows.

    %\todo[inline]{gg: We reduce from $k$-coloring, so we should use gadgets that steal $k-1$ colors. The current construction steals $k$ colors, which is not correct. dj: for camera-ready}

    \begin{enumerate}
        \item Create a \podcrit[(k+1)] triangle-free DAG $D$, which is guaranteed to exist by \cref{lem:podcrit-triangle-free}.
        \item Choose a sink $s$ and an arc $u \rightarrow s$ and delete the arc from $D$.
        \item Add an isolated vertex $v'$ to $D$.
        % \todo{gg: Add a disconnected arc $v' \rightarrow v''$ to $D$. We will use color of this arc as an indicator of the color of $v$ in the coloring of $G$.}
        \item Choose a \rbarc[(k-1)] $S \subseteq \fin(u)$, which exists by \cref{lem:rainbowarc-set-exists}.
        \item For each arc $(w \rightarrow  u)\in S$, add an arc $w \rightarrow  v'$ to $D$.
    \end{enumerate}
    See \cref{fig:rainbow} for an example. 
    By construction, the gadget is also a triangle-free DAG.

    \begin{claim}
        For any vertex gadget constructed as above, $\fin(v')$ is a \rbarc[(k-1)].
    \end{claim}

    \begin{claimproof}
To prove the claim, let $c$ be a proper $k$-coloring of
$UG(\linedg{D})$ corresponding to a partial order partition of $D$, and let $S \subseteq \fin(u)$ be the \rbarc[(k-1)] chosen in the
construction of the vertex gadget. No arc in $\fin(v')$ is adjacent to an arc
of $S$ in~$UG(\linedg{D})$. Moreover, for every $w$, the two arcs $(w \rightarrow v') \in \fin(v')$ and
$(w \rightarrow u) \in S$ have the same neighborhood outside $\fin(v') \cup S$, as~$u$ and~$v'$ are sinks. For~$u$, this is due to \cref{lem:u-sink}, and for~$v'$ it is by construction.
Consequently, we may recolor every arc $(w \rightarrow u) \in S$ with the color $c(w \rightarrow v')$,
and the resulting coloring is still proper. Since $S$ is a \rbarc[(k-1)], it
receives at least $k-1$ colors in the resulting coloring. Therefore, we have $|c(\fin(v'))| \geq k-1$.

Conversely, consider a proper $k$-coloring $c$ of $UG(\linedg{D}) \setminus \fin(v')$ in which $S$ receives exactly $k-1$ colors. Such a
coloring exists because $S$ is a \rbarc[(k-1)]. We can extend $c$ to $\fin(v')$ by
assigning $c(w \rightarrow u)$ to $w \rightarrow v'$ for every $(w \rightarrow v') \in \fin(v')$. The neighborhood property above, together
with the fact that there are no adjacencies between $\fin(v')$ and $S$,
ensures that this extension is valid. Thus, we have $|c(\fin(v'))|=|c(S)|=k-1$. Together with the first part of the proof, this shows that every proper $k$-coloring assigns at least $k-1$ colors to $\fin(v')$, and there exists a proper $k$-coloring assigning exactly $k-1$ colors to it. Hence $\fin(v')$ is a \rbarc[(k-1)], which proves the claim.
% \todo{gg: Use missing color to color arc $v' \rightarrow v''$.}
    \end{claimproof}

% We now construct our reduction as follows.
We now construct a triangle-free DAG $D$ from $G$ as follows; see \cref{fig:reduction} for an example.
\begin{enumerate}
    \item Start with an empty digraph $D$.
    \item For each vertex $v \in V(G)$, create a vertex gadget and add it to $D$. Furthermore, for each vertex $v \in V(G)$, add a vertex $v''$ and the arc $v' \rightarrow v''$ to $D$.
    \item Choose any order $\pi$ of the vertices of $G$. 
    \item Create an arc $v' \rightarrow w'$ between the two vertices $v'$ of $\mathcal V_v$ and $w'$ of $\mathcal V_w$ in $D$ if the corresponding vertices $v,w$ are adjacent in $G$ and $\pi(v) < \pi(w)$.
    % \item Add a vertex, $v_0$ to $D$.
    % \item For every vertex gadget, add the arc $v' \rightarrow v_0$ to $D$.
\end{enumerate}
% \todo{gg: Do not add $v_0$.}

Suppose first that $G$ admits a proper $k$-coloring $c$. We construct a $k$-arc-coloring $c'$ of $D$ corresponding to a partial order partition. For each vertex $v$ in $G$, we color each arc leaving the corresponding vertex $v'$ in $D$ with $c(v)$.
By the claim above, we may color each vertex gadget so that $\fin(v')$ receives exactly the $k-1$ colors different from $c(v)$. Suppose that an arc $w' \rightarrow v'$ receives the same color as the outgoing arcs of $v'$. By construction,~$c(v)=c(w)$, which contradicts the assumption that $c$ is a proper coloring of~$G$, as~$v$ and $w$ are adjacent in $G$. Hence, no $2$-path through $v'$ is monochromatic, and since every other $2$-path lies entirely within one vertex gadget, whose coloring is proper on $UG(\linedg{\mathcal V_v})$ by construction, $c'$ is a proper coloring of $UG(\linedg{D})$. By \cref{cor:chromaticpoppod}, this corresponds to a partial order partition of $D$ into at most $k$ parts.

Conversely, suppose that $\popn{D} \leq k$ with a corresponding $k$-arc-coloring $c'$. Each $v'$ has at least one outgoing arc, $v' \rightarrow v''$, so there is at least one color in $c'(\fout(v'))$. By the claim above, $\left| c'(\fin(v')) \right| \geq k-1$, so there is at most one color in $c'(\fout(v'))$. Thus, we can construct a coloring of the vertices of $G$ by letting $c(v)$ denote the unique color in $c'(\fout(v'))$. Assume that $c$ is not a proper coloring, meaning that there exists $\{v_1, v_2\} \in E(G)$ such that $c(v_1) = c(v_2)$.
Without loss of generality, we may assume that the arc $v_1' \rightarrow  v_2'$ is in $D$. We consider the directed path formed by the vertices $v_1', v_2', v_2''$ and observe that $c'(v_1' \rightarrow  v_2') = c(v_1) = c(v_2) = c'(v_2' \rightarrow  v_2'')$, a contradiction. This concludes the reduction.
\end{proof}

\begin{figure}[ht]
    \centering
    \includegraphics[page=5]{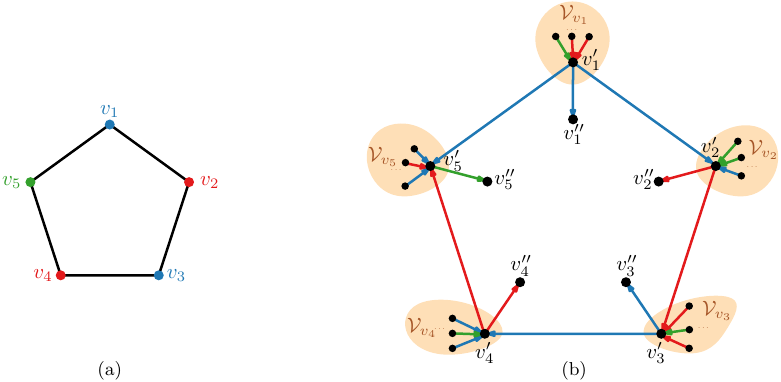}
    \caption{The reduction for $k=3$. (a) A $3$-coloring of $G=C_5$. (b) The corresponding DAG $D$, with shaded vertex gadgets and a sink $v''_i$ for each $v'_i$. For $i<j$ we orient the arc from $v'_i$ to $v'_j$.}
    \label{fig:reduction}
\end{figure}

Notice that our construction works for every fixed $k$, but we do not construct the vertex gadget explicitly and we do not have a polynomial bound in terms of $k$ on the size of a vertex gadget.
It would be interesting to have an explicit, polynomial-time construction of $k$-pop-critical digraphs.
\todo[disable]{gg: this is perhaps noy possible for DAGs, but it is possible for undirected graphs.}

    %\begin{figure}[ht]
    %    \centering
    %    \includegraphics[page=3]{figures/gadgetes.pdf}
    %    \caption{Coloring and Covering are different.}
    %    \label{fig:different}
    %\end{figure}
    
\section{Open Problems}
\label{sec:open}
While we were able to characterize the complexity of partial order decomposition on triangle-free DAGs, the complexity of this problem on tournaments remains open for $k \geq 3$.
Furthermore, we introduced the partial order partition number and showed that it can differ from the partial order decomposition number.
It is an intriguing question whether the difference between the two parameters can be arbitrarily large.

\bibliographystyle{plainurl}
\bibliography{abbrv,literature}

@inproceedings{schaefer_complexity_1978,
author = {Schaefer, Thomas J.},
title = {The complexity of satisfiability problems},
year = {1978},
publisher = {Association for Computing Machinery},
doi = {10.1145/800133.804350},
booktitle = {Proceedings of the Tenth Annual ACM Symposium on Theory of Computing},
pages = {216--226},
}

@misc{aboulker2026,
      title={Decomposing tournaments into comparability graphs}, 
      author={Pierre Aboulker and Logan Crew and Julien Duron and Xinyue Fan and Hugo Jacob and Rémy Kimbrough and Hidde Koerts and Benjamin Moore and Sophie Spirkl and Stéphan Thomassé},
      year={2026},
      url={https://arxiv.org/abs/2606.07748}, 
}

@article{larman1994ramsey,
  author  = {David Larman and Ji{\v r}{\'i} Matou{\v s}ek and J{\'a}nos Pach and Jen{\H o} T{\"o}r{\H o}csik},
  title   = {A {Ramsey}-Type Result for Convex Sets},
  journal = {Bulletin of the London Mathematical Society},
  volume  = {26},
  number  = {2},
  pages   = {132--136},
  year    = {1994},
  doi     = {10.1112/blms/26.2.132},
}

@misc{gyarfas2024partitioning,
  author        = {Andr{\'a}s Gy{\'a}rf{\'a}s and M{\'a}rton Marits and G{\'e}za T{\'o}th},
  title         = {Partitioning perfect graphs into comparability graphs},
  year          = {2024},
  eprint        = {2408.13523},
  archivePrefix = {arXiv},
  primaryClass  = {math.CO},
  doi           = {10.48550/arXiv.2408.13523},
  url           = {https://arxiv.org/abs/2408.13523},
}

@article{harner1972arc,
  author  = {C. C. Harner and R. C. Entringer},
  title   = {Arc Colorings of Digraphs},
  journal = {Journal of Combinatorial Theory, Series B},
  volume  = {13},
  number  = {3},
  pages   = {219--225},
  year    = {1972},
  doi     = {10.1016/0095-8956(72)90057-3},
}

@article{maffray1996,
author = {Frédéric Maffray and Myriam Preissmann},
title = {On the NP-completeness of the k-colorability problem for triangle-free graphs},
year = {1996},
journal = {Discrete Mathematics},
volume = {162},
number = {1},
pages = {313--317},
doi = {10.1016/S0012-365X(97)89267-9},
}

@inproceedings{erdoshajnal,
  author    = {Erd{\H{o}}s, Paul and Hajnal, Andr{\'a}s},
  title     = {On chromatic number of infinite graphs},
  year      = {1968},
  booktitle = {Theory of Graphs (Proceedings of the Colloquium held at Tihany, 1966)},
  pages     = {83--98},
}

@article{Mycielski1955,
  author  = {Mycielski, Jan},
  title   = {Sur le coloriage des graphes},
  journal = {Colloquium Mathematicum},
  volume  = {3},
  number  = {2},
  pages   = {161--162},
  year    = {1955},
  doi     = {10.4064/cm-3-2-161-162}
}

\end{document}